\documentclass[a4paper]{article}
\usepackage{amsfonts,amssymb,amsmath,amsthm,cite}
\usepackage{graphicx}
\usepackage{slashed}

\evensidemargin=\oddsidemargin
\newcommand{\PP}{\mathcal{P}}
\newcommand{\tr}{{\rm tr}} 
\newcommand{\ee}{{\bf e}}
\newcommand{\hee}{\hat{\bf e}}
\newcommand{\ff}{\hat{\bf f}}

\renewcommand{\i}{{\mathrm{i}}}

\newtheorem{assumption}{Assumption}[section]

\newtheorem{corollary}[assumption]{Corollary}

\newtheorem{lemma}[assumption]{Lemma}

\newtheorem{prop}[assumption]{Proposition}

\begin{document}

\begin{center}

{\Large\textbf{Spectral geometry of symplectic spinors}} \\
\vspace{0.5cm}

{\large Dmitri Vassilevich}

\vspace{0.5cm}

{\it Center of Mathematics, Computation and Cognition, Universidade Federal do ABC}\\
{\it 09210-580, Santo Andr\'e, SP, Brazil}

\end{center}

\begin{abstract}
Symplectic spinors form an infinite-rank vector bundle. Dirac operators on this bundle were constructed recently by K.~Habermann. Here we 
study the spectral geometry aspects of these operators. In particular, we define the associated distance function and compute the heat trace
asymptotics.
\end{abstract}

\section{Introduction}
According to the noncommutative geometry approach, geometry is defined by spectral triples. That is, geometry essentially becomes spectral geometry
of natural Dirac type operators. This approach is welcomed by physicists since it bridges up the differences between classical and quantum
geometries. This also moves geometry towards traditional areas of Mathematical Physics. Many links with Quantum Field Theory and even particle physics have been established, see \cite{CM}.
The spectral geometry of Riemannian manifolds has been studied in detail, while the spectral geometry of symplectic manifolds remains a largely uncharted area.

Quantization starts with a Poisson structure, that becomes a symplectic structure in the non-degenerate case. Therefore, from the point of view of Quantum Theory, symplectic manifolds are more important than the Riemannian ones. The purpose of this work is to extend the spectral 
geometry approach to the symplectic spinors. We address two important aspects, namely the spectral distance function and the heat trace 
asymptotics.

The symplectic spinors were introduced by B.~Kostant \cite{Kostant}. The Dirac operator on symplectic
spinors and the corresponding Laplacian were defined by K.~Habermann \cite{H1,H2}, who has also studied their basic
properties. A nice overview is the monograph \cite{HH}, whose conventions and notations we mostly follow in this work.
The works \cite{H1,H2} defined two Dirac operators, $\mathcal{D}$ and $\widetilde{\mathcal{D}}$. Since the symplectic spinor bundle has an infinite rank, neither of these two operators is a Dirac operator in the noncommutative geometry sense. Moreover, the relevant Laplacian $\mathcal{P}$ appears to be the commutator of $\mathcal{D}$ and $\widetilde{\mathcal{D}}$ (rather than the square of $\mathcal{D}$ or $\widetilde{\mathcal{D}}$). Therefore, the basic notions of spectral geometry cannot be immediately applied.

We shall show that despite the difficulties described above, there exists a modification of the standard distance formula of noncommutative
geometry that reproduces the geodesic distance on base manifold $M$ through symplectic Dirac operators. This is probably the most surprising
result of this work.

If the base manifold $M$ is almost hermitian, the operator $\PP$ leaves some subbundles $\mathbf{Q}^J_l$ of the symplectic spin bundle $\mathbf{Q}$ invariant. These subbundles are of finite rank.
Let $\PP_l$ be a restriction of Laplacian $\PP$ to $\Gamma (\mathbf{Q}^J_l)$, then $\PP_l$ is also a Laplace type operator. This will allow us
to develop the theory of heat trace asymptotic for $\PP_l$, identify corresponding invariants and compute a couple of leading terms in the
asymptotic expansion. Even more detailed information on the heat trace can be obtained when ${\rm dim}\, M=2$ and in the particular case
$M=CP^1$. We shall consider these cases as examples in the last section of this paper.

\section{Preliminaries}
In this Section we collect some basic facts that will be useful later. In what refers to symplectic spinors we mostly follow \cite{HH}.

\subsection{Differential geometry of almost hermitian manifolds}
Consider a symplectic manifold $(M,\omega)$, ${\rm dim}\, M=2n$,  with $\omega$ being a symplectic form, equipped with an almost complex structure $J$ and
a Riemannian metric $g$ related through $g(X,Y)=\omega(X,JY)$ for $X,Y \in TM$. Let the symplectic connection $\nabla$ 
be Hermitian, which implies
\begin{equation}
\nabla g=0,\qquad \nabla J=0\,.\label{ngnJ}
\end{equation}
The symplectic form $\omega$ is also covariantly constant, $\nabla\omega=0$. This makes $M$ a Fedosov manifold \cite{GRS}.
Almost hermitian manifolds admit unitary tangent frames in that both the metric and the symplectic form
have a canonical form \cite{KN}. Such frames form a principal $U(n)$ bundle $U(M)$ over $M$. We shall follow
the notations of \cite{HH} and write such frames as $(\ee_1,\dots, \ee_{2n})=(\hee_1,\dots,\hee_n,\ff_1,\dots,\ff_n)$.
They satisfy
\begin{eqnarray}
&&g(\ee_i,\ee_j)=\delta_{ij}\,,\nonumber\\
&&\omega(\hee_i,\ff_j)=\delta_{ij},\qquad \omega(\hee_i,\hee_j)=\omega(\ff_i,\ff_j)=0\,.
\label{gomef}
\end{eqnarray}
The almost complex structure $J$ acts on the basis as
\begin{equation}
J\hee_j=\ff_j,\qquad J\ff_j=-\hee_j\,.\label{Jef}
\end{equation}

It is interesting and useful to follow certain analogies with the Yang-Mills theory. 
By linearity, there is $u\in \Gamma(T^*M\times {\rm End}\,U(M))$ such that for any vector
field $X$
\begin{equation}
\nabla_X \ee_i = (Xu)\ee_i=X^\mu\, u_{\mu\, ij}\ee_j \,.\label{naXe}
\end{equation}
Whenever it cannot lead to a confusion we shall use the Einstein conventions of summation over repeated indices. 
The Greek indices $\mu,\nu,..$ are vector indices corresponding to a local coordinate chart. They are introduced to make connections with
the Yang-Mills more explicit.  
It is easy to check that $\nabla$ is a symplectic hermitian connection iff $u$ is antisymmetric and
commutes with $J$, i.e. iff $u$ belongs to the $2n$-dimensional real representation of $\mathfrak{u}(n)$.
Eq.\ (\ref{naXe}) allows to express the Christoffel symbol through the vectors of unitary frame 
$\ee_i$ and the $U(n)$ connection one-form $u$. 
One has the following expression for the torsion:
\begin{equation}
T(\ee_i,\ee_j)=(\ee_i u)\ee_j-(\ee_ju)\ee_i -[\ee_i,\ee_j]\,.\label{Tee}
\end{equation}
Also the curvature of $\nabla$ can be expressed through $u$:
\begin{equation}
(\mathcal{R}(\ee_i,\ee_j)\ee_k)^\rho =\ee_i^\nu\ee_j^\mu F_{\mu\nu km}\ee_m^\rho\,,\label{RF}
\end{equation}
where
\begin{equation}
F_{\mu\nu km}=-\partial_\mu u_{\nu km}+\partial_\nu u_{\mu km}+[u_\mu,u_\nu]_{km}\label{Fu}
\end{equation}
is the Yang-Mills type curvature associated to $u_\mu$. We remind that the generators of $\mathfrak{u}(n)$ algebra 
are labeled by the pairs of indices $(k,m)$.

The heat trace asymptotics of Laplace type operators are usually expressed in terms of the Levi-Civita
connection $\nabla^{\rm LC}$ and corresponding curvatures. This connection is related to $\nabla$ by the text-book 
formula:
\begin{equation}
g(\nabla_X Y, Z)=g(\nabla_X^{\rm LC} Y, Z) + \tfrac 12 \bigl[ g(T(X,Y),Z)-g(T(X,Z),Y) - g(T(Y,Z)X \bigr]
\label{LCtors}
\end{equation}
The following vector, $\mathfrak{T}$, and covector, $\tau$, fields are associated with the torsion:
\begin{equation}
\mathfrak{T}=\sum_{j=1}^n T(\hee_j,\ff_j)\,,\qquad \tau(X)=\sum_{k=1}^{2n} g(T(\ee_k,X),\ee_k) \,.\label{Ttau}
\end{equation}
There is a useful relation which involves the Riemann scalar curvature $\bar\rho$:
\begin{eqnarray}
&&\sum_{j,k=1}^{2n}g(R(\ee_j,\ee_k)\ee_j,\ee_k)=-\bar \rho + 2\sum_{j=1}^{2n} \nabla_{\ee_j}^{\rm LC} \tau(\ee_j)
+\sum_{j=1}^{2n} \tau(\ee_j)^2 \nonumber\\
&&\qquad\qquad -\frac 14 \sum_{j,l}^{2n}\bigl[ 2g(T(T(\ee_j,\ee_l),\ee_l),\ee_j)+g(T(\ee_j,\ee_l),T(\ee_j,\ee_l))\bigr] 
\,.\label{Rrho}
\end{eqnarray}

\subsection{Symplectic spinors and relevant operators}
Let us remind that the metaplectic group $Mp(2n,\mathbb{R})$ is a two-fold covering of the symplectic group
$Sp(2n,\mathbb{R})$. A metaplectic structure is a principal $Mp(2n,\mathbb{R})$ bundle (together with a morphism 
that ensures consistency with the symplectic structure). The metaplectic group is represented in the space of square integrable
functions $L^2(\mathbb{R}^n)$. By using this representation one can define a bundle associated with the metaplectic structure, 
which is the symplectic spinor bundle $\mathbf{Q}$. For the sections of this fiber bundle, $\varphi,\psi\in\Gamma(\mathbf{Q})$,
one has a fiber scalar product, $\langle \varphi,\psi\rangle_x$, $x\in M$, and a fiber norm $\| \varphi \|_x^2=\langle\varphi,\varphi\rangle_x$.
By integrating $\langle \varphi,\psi\rangle_x$ over $M$ one obtains a scalar product on $\Gamma(\mathbf{Q})$.

For $X\in \Gamma(TM)$ and $\varphi\in\Gamma(\mathbf{Q})$ one defines the symplectic Clifford multiplication
$X\cdot \varphi$ satisfying
\begin{equation}
(X\cdot Y -Y\cdot X)\cdot \varphi =-\i\omega(X,Y)\varphi  \label{Clif}
\end{equation}
From now on, $\i \equiv \sqrt{-1}$.

The symplectic spinor bundle $\mathbf{Q}$ splits into an orthogonal sum of finite rank subbundles $\mathbf{Q}_l^J$, $l=0,1,2,\dots$,
\begin{equation}
{\rm rank}\, \mathbf{Q}^J_l=\frac {(l+n-1)!}{l!(n-1)!}\,. \label{rank}
\end{equation}
There is an important operator, $\mathcal{H}^J$, acting on $\varphi\in\Gamma(\mathbf{Q})$
\begin{equation}
\mathcal{H}^J\varphi = \frac 12 \sum_{j=1}^{2n} \ee_j \cdot \ee_j \cdot \varphi\,. \label{HJ}
\end{equation}
The operator $\mathcal{H}^J$ equals to a constant $q_l$ on each of $\mathbf{Q}_l^J$, and
\begin{equation}
q_l=-\bigl( l +\tfrac n2 \bigr) \label{ql}\,.
\end{equation}

In the quantum mechanical language, the Clifford multiplication by $\hee_j$ may be thought of as a canonical coordinate operator,
while $\ff_j\cdot$ may be thought of as a conjugate momentum operator. Then $\mathcal{H}^J$ is minus the Hamiltonian of $n$-dimensional
harmonic oscillator, while the ladder operators are
\begin{equation}
	L^{(\pm)}_j \varphi = (\ff_j \mp \i \hee_j)\cdot \varphi \,.\label{Lpm}
\end{equation}
For any $\varphi_0\in \Gamma(\mathbf{Q}_0^J)$
\begin{equation}
	L^{(+)}\varphi_0=0\,.\label{Lplus}
\end{equation}

Any symplectic connection $\nabla$ on $M$ induces a covariant derivative on $\mathbf{Q}$ that will be denoted by the same letter $\nabla$.
We shall need the relation
\begin{equation}
	\nabla_X (Y \cdot \varphi) = \nabla_X Y \cdot \varphi + Y\cdot \nabla_X \varphi \label{scovd}
\end{equation}
between the derivative and the Clifford multiplication and the corresponding curvature
\begin{eqnarray}
\mathcal{R}^{\mathbf{Q}}(X,Y)\varphi& \equiv & \nabla_X\nabla_Y\varphi - \nabla_Y\nabla_X\varphi -\nabla_{[X,Y]}\varphi
\nonumber\\
&=& -\frac i2 \sum_{j=1}^{2n} \mathcal{R}(X,Y)\ee_j \cdot J\ee_j \cdot \varphi \,. \label{curvaQ}
\end{eqnarray}

Given a unitary frame $(\ee_1,\dots,\ee_{2n})$ one defines a pair of Dirac operators
\begin{equation}
\mathcal{D}\varphi = -\sum_{j=1}^{2n} J\ee_j\cdot \nabla_{\ee_j}\varphi\,,\qquad
\widetilde{\mathcal{D}}\varphi = \sum_{j=1}^{2n} \ee_j\cdot \nabla_{\ee_j}\varphi\,.\label{DD}
\end{equation}
Note, that the formula for $\mathcal{D}$ may be rewritten in the form
\begin{equation}
\mathcal{D}\varphi =\sum_{j=1}^n \bigl( \hee_j \cdot \nabla_{\ff_j}\varphi -\ff_j \cdot \nabla_{\hee_j}\varphi \bigr)
\label{Dother}
\end{equation}
which does not use the metric or the almost complex structure. 

An associated second order operator is defined through the commutator 
\begin{equation}
\mathcal{P}=\i [\widetilde{\mathcal{D}},\mathcal{D}] \,.\label{P}
\end{equation}
The principal symbol of $\mathcal{P}$ is given by the inverse metric. The operator $\mathcal{P}$ leaves the subbundles
$\mathbf{Q}_l^J$ invariant, $\mathcal{P}:\Gamma (\mathbf{Q}_l^J)\to \Gamma (\mathbf{Q}_l^J)$, though neither $\widetilde{\mathcal{D}}$ 
nor $\mathcal{D}$ have this property. A restriction of $\mathcal{P}$ to $\Gamma(\mathbf{Q}_l^J)$ will be denoted by $\mathcal{P}_l$

\subsection{Heat trace asymptotics}
We shall need some facts from the theory of asymptotic expansion of the heat trace
associated with Laplacians \cite{GilkeyNew} (see also \cite{Kirsten,Vassilevich:2003xt}).

Let $V$ be a finite rank vector bundle over a compact Riemannian manifold $M$ without boundary. Let $P$ be
a Laplace type operator on $\Gamma (V)$. Then\\
(1) exist a unique endomorphism $\bar E$ and a unique connection $\bar \nabla$ of $V$ such that
\begin{equation}
P=-(\bar\nabla^2+\bar E)\,,\label{nabE}
\end{equation} 
where the square in $\bar \nabla^2$ is calculated with the Riemannian metric on $M$ and includes the 
metric Christoffel symbol.\\
(2) The heat trace
\begin{equation}
K(P,t):={\rm Tr}\, \bigl( \exp(-tP) \bigr)\,, t\in\mathbb{R}_+ \label{KPP}
\end{equation}
exists and admits a full asymptotic expansion 
\begin{equation}
K(P,t)\simeq \sum_{k=0}^\infty t^{k-n} a_{2k}(P) \label{asymptotex}
\end{equation}
as $t\to +0$. Here $2n={\rm dim}\, M$.\\
(3) The heat trace coefficients $a_{2k}$ are locally computable. Namely, each $a_{2k}$ is given by the integral
over $M$ of the bundle trace of a local invariant polynomial constructed from the endomorphism $\bar E$,
from the Riemann tensor $\bar R_{ijkl}$, the curvature $\bar \Omega_{ij}$ of $\bar \nabla$, and their
derivatives. In particular,
\begin{eqnarray}
&&a_0(P)=(4\pi)^{-n} \int_M \tr_V (I)\,,\label{a0}\\
&&a_2(P)=(4\pi)^{-n}\,\tfrac 16 \int_M \tr_V \bigl( 6\bar E+\bar \rho \bigr) \,,\label{a2}\\
&&a_4(P)=(4\pi)^{-n} \, \tfrac 1{360} \int_M \tr_V \bigl( 60\bar \rho \bar E + 180 \bar E^2+ 5\bar \rho^2\nonumber\\
&&\qquad\qquad\qquad -2\overline{\rm Ric}_{ij}\overline{Ric}_{ij} 
+ 2\bar R_{ijkl}\bar R_{ijkl}  + 30\bar \Omega_{ij}\bar\Omega_{ij} \bigr)\,. \label{a4}
\end{eqnarray}
The Ricci tensor $\overline{Ric}_{ij}=\bar R_{lijl}$ and the scalar curvature $\bar \rho=\overline{Ric}_{jj}$
are defined in such a way that $\bar \rho=2$ on the unit $S^2$. Summation over the repeated indices is understood.

Note, that the expansion (\ref{asymptotex}) contains even-numbered coefficients and integer powers of $t$ only.
Odd-numbered coefficients appear e.g. on manifolds with boundaries.

Our purpose is to calculate the coefficients $a_{2k}$ and relate them to geometric
invariants of the symplectic spinor bundle.

\section{The distance}
Let us start with basic definitions related to the distance function in noncommutative geometry, see \cite{Varilly} for a brief introduction. 
Consider a spectral triple $(\mathcal{A},H,D)$ consisting of an algebra $\mathcal{A}$, acting on a Hilbert space $H$ by bounded
operators, and of a Dirac operator $D$. The commutator $[D,a]$ has to be bounded for all $a\in \mathcal{A}$ (or at least the set of $a$ for which $[D,a]$ is bounded has to be dense in $\mathcal{A}$). One can define a distance
between two states $x$ and $y$ on the algebra $\mathcal{A}$ by the formula \cite{CDist}
\begin{equation}
d(x,y)=\sup_{a\in\mathcal{A}} \{ |a(x)-a(y)|: \| [D,a]\| \le 1\} \,.\label{specdist}
\end{equation}

For a commutative spectral triple, $\mathcal{A}$ is the algebra $C(M)$ of continuous functions on a compact Riemannian  spin$^c$ manifold $M$,
$H$ is the space of square-integrable spinors, and $D$ is the canonical Dirac operator. The pure states $x$ correspond to points on $M$
with $x:a\mapsto a(x)$ being the evaluation map.  The distance (\ref{specdist}) coincides with the geodesic distance on $M$. 
Since $D$ is the canonical Dirac operator
	\begin{equation}
	[D,a]=-\i \gamma (da)\,,\label{Da}
\end{equation}
where $\gamma$ is a composition of the orthogonal Clifford multiplication and isomorphism between tangent and cotangent bundles
made with the inverse Riemannian metric $g^{-1}$. In short, for two one forms $\alpha$ and $\beta$, $\gamma(\alpha)\gamma(\beta)+
\gamma(\beta)\gamma(\alpha)=2g^{-1}(\alpha,\beta)$. Fiberwise, $\gamma(\alpha)$ is a hermitian matrix. The restriction
$\| [D,a]\| \le 1$ is equivalent to
\begin{equation}
	|g^{-1}(da^*,da)|\le 1\,.\label{gdada}
\end{equation}
Complexity or reality of the function $a$ plays no role here. We shall consider real functions in what follows.

In the context of symplectic spinors one may take the same $\mathcal{A}=C(M)$, the Hilbert space may be formed by sections of the symplectic
spinor bundle. However, the commutators $[\mathcal{D},a]$ and $[\widetilde{\mathcal{D}},a]$ are practically 
never bounded, so that Eq.\ (\ref{specdist})
with $\mathcal{D}$ or $\widetilde{\mathcal{D}}$ instead of $D$ does not define any interesting distance on $M$. Therefore, we return again to
the usual spin case and replace the condition $\| [D,a] \|\le 1$ by an equivalent one, which however can be naturally generalized for symplectic spinors.
Let us take a section $\psi_0$ of
spin bundle that has a unit fiber norm at each point of $M$: $1=\| \psi_0 \|_x^2 = \psi^\dag_0 (x)\psi_0(x)$.  Let us take a real
$a=a^*$ in $\mathcal{A}$ and compute
\begin{equation*}
\| [D,a] \psi_0\|_x^2=\psi_0^\dag(x) \gamma(da)\gamma(da) \psi(x) = g^{-1}(da,da) \psi_0^\dag(x)\psi_0(x)=
g^{-1}(da,da)\,.
\end{equation*}
This shows that the condition $\| [D,a] \psi_0\|_x \le 1$ for all points $x$ on $M$ may be used instead of the original
restriction on the norm $\| [D,a] \| \le 1$ to define the distance on $M$.

A similar construction in the symplectic case goes as follows.
Take a section $\varphi_0$ of $\mathbf{Q}^J_0$ with a constant fiber norm $\| \varphi_0\|^2_x=2$ everywhere on $M$. 
Then, 
\begin{eqnarray}
[\widetilde{\mathcal{D}},a]\varphi_0 &=& \sum_{j=1}^n \left( (\hee_j a)\hee_j \cdot \varphi_0 + (\ff_j a) \ff_j \cdot \varphi_0
\right) = \sum_{j=1}^n \left( (\hee_j a) \bigl( -\tfrac {\i}2 L_j^{(-)} \bigr)  + (\ff_j a) \tfrac 12 L_j^{(-)} \right) \varphi_0
\nonumber\\
&=& \sum_{j=1}^n \left( -\i (\hee_j a)  + (\ff_j a)  \right)\tfrac 12 L_j^{(-)} \varphi_0\,.
\end{eqnarray}
Furthermore, 
\begin{eqnarray}
&&(\tfrac 12 L_j^{(-)}\varphi_0,\tfrac 12 L_k^{(-)}\varphi_0)_x=\tfrac 14 (\varphi_0,L_j^{(+)} L_k^{(-)}\varphi_0)_x=
\tfrac 14 (\varphi_0,(L_j^{(-)} L_j^{(+)}+2)\varphi_0)_x \delta_{jk}\nonumber\\
&&\qquad =\tfrac 12 (\varphi_0,\varphi_0)_x \delta_{jk}=\delta_{jk}\,,
\end{eqnarray}
where we used Eq. (\ref{Lplus}). Finally,
\begin{equation}
\| [\widetilde{\mathcal{D}},a]\varphi_0 \|_x^2 = \sum_{j=1}^{2n} (\ee_j a)^2 = g^{-1}(da,da)\,.\label{normDa}
\end{equation}

By collecting everything together we arrive at the following

\begin{prop}
Let $\mathcal{A}$, $\varphi_0$ and $\widetilde{\mathcal{D}}$ be as defined above. Then
\begin{equation}
d(x,y)=\sup_{a\in\mathcal{A}} \{ |a(x)-a(y)|: \| [\widetilde{\mathcal{D}},a]\varphi_0 \|_p \le 1, \forall p\in M \} \,.\label{tilDdist}
\end{equation}
is the geodesic distance between two points $x$ and $y$ on $M$.
\end{prop}

\paragraph{Remark} The same distance is obtained if one uses everywhere $\mathcal{D}$ instead of $\widetilde{\mathcal{D}}$, which can be easily verified. However, this is true only in the rather restrictive setup used in this work. In general, one does not need a Riemannian
metric to define $\mathcal{D}$ \cite{H1,H2,HH}. It is enough to take a symplectic manifold with a metaplectic structure and a symplectic
connection. Then, $(\ee_1,\dots,\ee_{2n})=(\hee_1,\dots,\hee_n,\ff_1,\dots,\ff_n)$ in Eq.\ (\ref{Dother}) may be any symplectic frame. The
distance (\ref{tilDdist}) (with $\mathcal{D}$ in place of  $\widetilde{\mathcal{D}}$) \emph{defines} a metric on $M$. 
\section{Heat trace asymptotics for $\mathcal{P}_l$. Generic case}
To compute the heat trace asymptotics one has to rewrite $\mathcal{P}_l$ in the canonical form (\ref{nabE}), compute the corresponding
invariants and calculate the traces. 
Our starting point is the Weitzenb\"ock formula \cite{HH} for $\mathcal{P}$ , 
\begin{equation}
\mathcal{P}\varphi =\nabla^* \nabla \varphi
 +\i \sum_{jk}^{2n} J\ee_j \cdot \ee_k \cdot \bigl( \mathcal{R}^{\bf Q} (\ee_j,\ee_k) \varphi
-\nabla_{T(\ee_j,\ee_k)} \varphi \bigr) - \nabla_{J\mathfrak{T}}\varphi \,. \label{Pdiv}
\end{equation}
Here $\nabla^*:\Gamma(T^*M\otimes \mathbf{Q})\to \Gamma(\mathbf{Q})$ is the formal adjoint operator of the spinor covariant derivative
$\nabla$. It is easy to see that $\nabla^*\nabla = -\nabla^2$ in the notations of Eq.\ (\ref{nabE}).
After some algebra we obtain the quantities appearing in the canonical form (\ref{nabE}) of the Laplacian
\begin{eqnarray}
&& \bar\nabla_X\varphi =\nabla_X \varphi + g(X,v)\varphi \label{barnab}\\
&&v= \tfrac {\i}2 \sum_{j,k=1}^{2n} T(\ee_j,\ee_k)\, J\ee_j \cdot \ee_k \cdot
+\tfrac 12 J\mathfrak{T} \label{v}\\
&&\bar E= -{\rm div}^{\rm LC}\, v - g(v,v) -\i \sum_{jk}^{2n} J\ee_j \cdot \ee_k \cdot \mathcal{R}^{\bf Q} (\ee_j,\ee_k)
\label{barE}
\end{eqnarray}
In the last formula the divergence corresponds to the Levi-Civita connection. In local terms, 
${\rm div}^{\rm LC}\, v =g^{\mu\nu} \nabla_\mu^{\rm LC} v_\nu$.

To calculate the traces we have to define the $\mathfrak{u}(n)$ representations corresponding to the objects appearing the formulas above.
Just looking at the expression (\ref{curvaQ}) for the curvature and at the relations (\ref{RF}), (\ref{Fu}) to the Yang-Mills field strength one
may conjecture that the combinations $\ee_j\cdot J\ee_k\cdot$ are generators of $\mathfrak{u}(n)$. Let us show that this is indeed the case.

Let $A\in\mathfrak{u}(n)$ be given by a matrix $A_{jk}$ in the $2n$-dimensional real defining representation.
This means that $A_{jk}$ is a real antisymmetric $2n\times 2n$ matrix which satisfies $A_{jk}=A_{j+n,k+n}$
and $A_{j,k+n}=-A_{j+n,k}$ for $j,k\leq n$. Then
\begin{equation}
r_Q(A)=\tfrac {\i}2 \sum_{j,k=1}^{2n} A_{jk}\ee_k\cdot J\ee_j\cdot \label{rQ}
\end{equation}
is a representation of $A$ on the symplectic spinors. Indeed,
\begin{eqnarray}
&&[r_Q(A),r_Q(B)]\varphi=\left( \tfrac {\i}2 \right)^2
[A_{jk}\ee_k\cdot J\ee_j,B_{lm}\ee_m\cdot J\ee_l]\cdot \varphi\nonumber\\
&&\ =\tfrac {\i}2 [A,B]_{jk} \ee_k\cdot J\ee_j \cdot \varphi =r_Q([A,B])\varphi \,.
\label{XuYu}
\end{eqnarray}
This representation is, of course, reducible. The action is fiberwise, so that it is enough to understand
$r_Q$ on each fiber.

\begin{lemma}
On the fibers of $\mathbf{Q}_l^J$ the representation $r_Q$ is equivalent to an $\mathfrak{su}(n)$ representation
with the Dynkin indices $(l,0,\dots,0)$, while the $\mathfrak{u}(1)$ charge is $q_l$.\label{Lrep}
\end{lemma}
\begin{proof}
In the $2n$-dimensional real defining representation non-zero matrix elements of 
the Cartan generators of $\mathfrak{u}(n)$ are
$(K_j)_{j,j+n}=-(K_j)_{j+n,j}=1$, $j=1,\dots,n$. Therefore,
\begin{equation}
r_Q(K_j)\varphi =\tfrac {\i}2 (\hee_j \cdot \hee_j + \ff_j\cdot \ff_j)\cdot\varphi \label{rQKj}
\end{equation}
Locally, a fiber of $\mathbf{Q}^J_l$ can be viewed as a linear space spanned by products
$h_{\alpha_1}(x_1)\cdots h_{\alpha_n}(x_n)$ of the Hermite functions $h_{\alpha}(x)$
with $\alpha_1+ \dots +\alpha_n=l$ and with Clifford multiplications by $\hee_j$ and $\ff_j$ represented by $ix_j$ and
$\partial_j$, respectively \cite{HH}. Hence, $r_{Q}(K_j)h_{\alpha_j}(x_j)=-\i (\alpha_j+\tfrac 12)h_{\alpha_j}(x_j)$.
Let us take the Cartan generators corresponding to ordered positive simple roots of $\mathfrak{su}(n)$ as
$K_1-K_2$, $K_2-K_3$, ..., $K_{n-1}-K_{n}$. By calculating the eigenvalues of this generators on
$h_{l}(x_1)h_0(x_2)\cdots h_0(x_n)$, we conclude that this monomial is the highest weight vector
of the representation $(l,0,\dots,0)$. To demonstrate that a fiber of $\mathbf{Q}^J_l$ is just this
representation and nothing else, it is enough to compare the dimensions (cf \cite{HH}),
\begin{equation}
{\rm dim}\, (l,0,\dots,0)=\frac {(l+n-1)!}{l!(n-1)!}={\rm rank}\, \mathbf{Q}^J_l \,.\label{dimRep}
\end{equation}
The $\mathfrak{u}(1)$ charge is, up to the imaginary unit, the eigenvalue of the $\mathfrak{u}(1)$
generator
\begin{equation}
\sum_{j=1}^n r_Q(K_j)=i\mathcal{H}^J\,,
\end{equation}
which is $iq_l$.
\end{proof}

Let us define a projector $\Pi$ on the $\mathfrak{u}(1)$ generator 
\begin{equation}
\Pi\, r_Q(A)=\frac{\i \mathcal{H}^J}n \sum_{j=1}^n A_{j,j+n} \label{Pi}
\end{equation}
and a pull-back of $\Pi$ to the $2n$-dimensional real representation, $\Pi\, r_Q(A)=r_Q(\Pi A)$.

\begin{corollary}
\begin{eqnarray}
&&{\rm tr}_l \big( r_Q(A) \big) = \frac{\i q_l\, {\rm rank}\, \mathbf{Q}^J_l}n \sum_{j=1}^n A_{j,j+n}\\
&&{\rm tr}_l \big( \Pi r_Q(A)\Pi r_Q(B) \big) = -\frac{q_l^2\, {\rm rank}\, \mathbf{Q}^J_l}{n^2} \sum_{j,k=1}^n A_{j,j+n}B_{k,k+n}\\
&&{\rm tr}_l \big( (1-\Pi)r_Q(A) (1-\Pi) r_Q(B)\big) =\frac{(l+n)!}{4(l-1)!(n+1)!} \, {\rm tr}\, \big( (1-\Pi)A(1-\Pi)B\big)\label{Cor3}
\end{eqnarray}
\end{corollary}
\begin{proof}
The first two lines above immediately follow from the fact that on each $\mathbf{Q}_l^J$ the operator $\mathcal{H}$ is an identity matrix
times $q_l$. To get the last line one uses that the trace forms in irreducible representations of a simple Lie algebra are proportional
between themselves. I.e., for $A$ and $B$ being matrices in the $2n$ dimensional real representation of $\mathfrak{su}(n)$
\begin{equation}
{\rm tr}_l \big( r_Q(A) r_Q(B) \big) = c(n,l)\, {\rm tr}\, \big(AB\big)\,.\label{trAB}
\end{equation}
Next, we take the quadratic Casimir element $C_2$ and calculate its' trace once by using the relation above, and then by using the
fact that $C_2$ in any irreducible representation is a unit matrix times the eigenvalues $C_2(n,l)$,
\begin{equation}
{\rm tr}_l\, C_2=\mu_n (n^2-1)\, c(n,l)={\rm rank}\, \mathbf{Q}^J_l C_2(n,l)\,,
\end{equation}
where $n^2-1={\rm dim}\, \mathfrak{su}(n)$ and $\mu_n$ is a normalization coefficient. Hence,
\begin{equation}
	c(n,l)=c(n,1) \, \frac {{\rm rank}\, \mathbf{Q}^J_l C_2(n,l)}{{\rm rank}\, \mathbf{Q}^J_1 C_2(n,1)} \,.
\end{equation}
The value $c(n,1)=\tfrac 14$ can be easily recovered by considering the realization on products of Hermite functions,
$C_2(n,l)= l(n+l)(n-1)/n$ is to be found in any textbook, see \cite{BR} for instance, and dimensions of relevant representations have been given
above. By collecting everything together, one arrives at (\ref{Cor3}).
\end{proof}

Explicitly, 
\begin{equation}
{\rm tr}\, \big( (1-\Pi)A(1-\Pi)B\big)=\sum_{j,k=1}^{2n} A_{jk}B_{kj}+\frac 2n \sum_{i,l=1}^n A_{i,i+n}B_{l,l+n} \,.\label{tr1Pi}
\end{equation}

With these formulas,we calculate traces of the terms appearing in Eq.\ (\ref{barE}) for $\bar E$
\begin{eqnarray}
&&{\rm tr}_l\, g(v,v)={\rm rank}\, \mathbf{Q}^J_l \left( \frac 14 -\frac{q_l^2}{n^2} +\frac{(l+n)l}{2n^2(n+1} \right)\, 
g(\mathfrak{T},\mathfrak{T}) \nonumber\\
&&\qquad\qquad - \frac {(l+n)!}{4(l-1)!(n+1)!} \sum_{j,k=1}^{2n} g(T(\ee_j,\ee_k),T(\ee_j,\ee_k))\,,\label{trl}\\
&&{\rm tr}_l\, \bigl( -\i \sum_{jk=1}^{2n} J\ee_j \cdot \ee_k \cdot \mathcal{R}^{\mathbf{Q}}(\ee_j,\ee_k) \bigr) \nonumber\\
&&\qquad\qquad =\left( - \frac{2q_l^2\, {\rm rank}\, \mathbf{Q}^J_l}{n^2} +\frac{(l+n)!}{n(l-1)!(n+1)!}\right)
\sum_{i,l=1}^n g(R(\ee_i,\ee_{i+n})\ee_l,\ee_{l+n})
\nonumber\\
&&\qquad\qquad -\frac{(l+n)!}{2(l-1)!(n+1)!} \sum_{j,k=1}^{2n} g(R(\ee_j,\ee_k)\ee_j,\ee_k) \,.
\nonumber
\end{eqnarray}

Let us introduce a short hand notation
\begin{equation*}
\alpha(n,l)\equiv \frac{l(l+n)}{n(n+1)}
\end{equation*}
Then the first two heat trace coefficients read
\begin{eqnarray}
&&a_0(\mathcal{P}_l) = (4\pi)^{-n} {\rm rank}\, \mathbf{Q}^J_l \, {\rm vol}\, (M) \label{a0Pl}\\
&&a_2(\mathcal{P}_l) = (4\pi)^{-n} {\rm rank}\, \mathbf{Q}^J_l \, \int_M \Bigl[ \Bigl( \frac 16 + \frac{ \alpha(n,l)}2 \Bigr) \bar\rho
+\Bigl( \alpha(n,l) -\frac{2q_l^2}{n^2} \Bigr)\sum_{i,j=1}^n g(R(\hee_i,\ff_{i})\hee_j,\ff_{j}) \nonumber\\
&&\qquad\qquad
 +\alpha (n,l) \sum_{i,j=1}^{2n} \Bigl( \frac 14 g(T(T(\ee_i,\ee_j),\ee_i),\ee_j) +\frac 38 g(T(\ee_i,\ee_j),T(\ee_i,\ee_j))\Bigr)
\nonumber\\
&&\qquad\qquad +\Bigl( -\frac 14 +\frac{q_l^2}{n^2}-\frac{\alpha(n,l)}{2n} \Bigr)\, g(\mathfrak{T},\mathfrak{T}) 
-\frac 12 \alpha(n,l) \sum_{j=1}^{2n} \tau (\ee_j)^2 \Bigr]\label{a2Pl}
\end{eqnarray}

\paragraph{Remarks.} 1. The coefficients (\ref{a0Pl}) and (\ref{a2Pl}) carry no dependence on the metaplectic structure. The structure of
invariants (\ref{barnab}) - (\ref{barE}) tells us that higher heat kernel coefficients have the same property. In the language of Kac 
\cite{Kac} this means that through the heat trace expansion of $\mathcal{P}_l$ one can hear the shape of symplectic almost hermitian 
manifolds, but not of the metaplectic structures.\\
2. In general, the operator $\mathcal{P}$ is not self-adjoint. However, the coefficients (\ref{a0Pl}) and (\ref{a2Pl}) are
real. We do not expect this property to hold for higher terms in the heat trace expansion.\\
3. One may define a family of spectral actions $S_l={\rm Tr}\, \big( \chi (\mathcal{P}_l /\Lambda^2\big)$ similarly to \cite{CC} with a 
cut-off function $\chi$ and a scale parameter $\Lambda$. The large $\Lambda$ expansion of $S_l$ is given by the heat trace asymptotics those
structure differs considerably for the standard case of spin Dirac operators with torsion, cf. \cite{HPS,ILV}. Note that in four dimensions
the spectral action for the spin Dirac operator is restricted by the chiral symmetry \cite{ILV}, that is not present in the symplectic case.

\section{Examples: two dimensions and $CP^1$}
In two dimensions any manifold admits a K\"ahler structure. To use all advantages of this low-dimensional case, we shall assume that 
$M$ is a K\"ahler manifold, so that the torsion vanishes. This reduces considerably the combinatorial complexity of the heat trace asymptotic expansion  and will allow to compute more heat trace coefficients. We also like to mention compact expressions \cite{LH} for the heat trace asymptotics of the scalar Laplace operator on K\"ahler manifolds.

In two dimensions, the expression for curvature simplifies
\begin{equation}
\mathcal{R}^{\mathbf{Q}}(X,Y)\varphi = -\i (R(X,Y)\hee,\ff) \mathcal{H}^J\varphi \,,\label{2dRQR}
\end{equation}
where $\hee\equiv \hee_1$, $\ff\equiv\ff_1$. For vanishing torsion
\begin{equation}
\PP \varphi = \nabla^*\nabla\varphi + \i\sum_{j,k=1}^2 J\ee_j \cdot \ee_k \cdot
\mathcal{R}^{\mathbf{Q}}(\ee_j,\ee_k)\varphi 
= -\nabla^2 \varphi - \rho \bigl(\mathcal{H}^J\bigr)^2\varphi \,,\label{2DW}
\end{equation}
where $\rho=\bar\rho$ is the scalar curvature.

Let us consider the operator $\PP_l$. We remind that
$\mathcal{H}^J\vert_{\mathbf{Q}^J_l}=-(l+\tfrac 12)\equiv q_l$. The heat trace asymptotics are characterized by the following
\begin{prop}
The operator $\PP_l$ has the form (\ref{nabE}) with $\bar\nabla=\nabla$,
\begin{equation}
E=\rho q_l^2,\qquad \Omega_{ij}= \tfrac {\i}2\, q_l \rho \omega_{ij} \,.\label{2DEOm}
\end{equation}
The heat kernel coefficients read
\begin{eqnarray}
&&a_0(\PP_l)=(4\pi)^{-1} {\rm vol}\ M\,,\label{2Da0}\\
&&a_2(\PP_l)=(4\pi)^{-1}\,\tfrac 16 (1+6q_l^2) \int_M \rho \,,\label{2Da2}\\
&&a_4(\PP_l)=(4\pi)^{-1}\,\tfrac 1{120} (2+15q_l^2+60q_l^4) \int_M \rho^2 \,.\label{2Da4}
\end{eqnarray}
\end{prop}

\begin{proof}
First, we recall that in two dimensions there is only one indepedent component of the Riemann tensor,
so that if $(R(\hee,\ff)\hee,\ff)=r$, then $\overline{Ric}(\hee,\hee)=\overline{Ric}(\ff,\ff)=-r$ and
$\bar\rho=\rho=-2r$. Then, $\bar\nabla=\nabla$ by inspection, so that the corresponding curvature is 
just $\mathcal{R}^{\mathbf{Q}}$. By Eq.\ (\ref{2dRQR}), we have
\begin{equation*}
\Omega (\hee,\ff)\varphi =\tfrac {\i}{2} \rho \mathcal{H}^J\omega(\hee,\ff)\varphi\,,
\end{equation*}
that yields the 2nd equation in (\ref{2DEOm}). The first equation there follows from (\ref{2DW}). Substitutions in 
(\ref{a0}) -- (\ref{a4}) lead to the desired result. 
\end{proof}

Let us restrict our attention further by taking $M=CP^1$. Then the eigenvalues $\lambda_{l,j}$ of $\PP_l$ and their degeneracies $m_{l,j}$
read (see \cite[Proposition 6.3.5]{HH})
\begin{equation}
\lambda_{l,j}=4(l+j+1)^2-3(2l+1)^2-1\,,\qquad m_{l,j}=2(l+j+1)\,,\qquad j=0,1,2,\dots\label{lm}
\end{equation}
To calculate the heat trace asymptotics one has  to evaluate the asymptotic expansion for
\begin{equation}
K(\PP_l,t)=\sum_{j=0}^\infty m_{l,j}\, e^{-t\lambda_{l,j}}
= e^{t(3(2l+1)^2+1)} \sum_{k=l+1}^\infty 2k\, e^{-4tk^2} \,.\label{CP11}
\end{equation}
To this end one may use the Euler-Maclaurin formula
\begin{equation}
\sum_{k=m}^q f(k)=\int_m^qf(x)dx +\frac 12 (f(q)+f(m))
+\sum_{i=1}^\infty \frac{B_{2i}}{(2i)!} \bigl(f^{(2i-1)}(q)-f^{(2i-1)}(m)\bigr) \label{EM}
\end{equation}
with $f(x)=2xe^{-4tx^2}$. Since $f^{(2i-1)}=O(t^{i-1})$, only a finite number of terms on the 
right hand side of (\ref{EM}) contribute to any given order of the expansion. Here 
$B_2=\tfrac 16,\ B_4=-\tfrac 1{30}, \dots$ are the Bernoulli numbers.
We have
\begin{equation}
\sum_{k=m}^\infty 2k\, e^{-4tk^2} \simeq e^{-4m^2t} \Bigl( \frac 1{4t} 
+\bigl( m-\frac 16 \bigr) + t \bigl( \frac 43 m^2 - \frac 1{15} \bigr) +O(t^2) \Bigr) \,.\label{CP12}
\end{equation}
Consequently,
\begin{equation}
K(\PP_l,t)\simeq \frac 1{4t} +\Bigl( \frac 56 -2m+2m^2\Bigr) 
+t\Bigl( \frac {19}{15} -6m +14m^2-16m^3 +8m^4 \Bigr) +O(t^2)\,,\label{CP13}
\end{equation}
where $m=l+1$. This expansion gives the values of $a_0$, $a_2$ and $a_4$ for $M=CP^1$.

The $CP^1$ with Fubini-Study metric is isometric to $S^2$ with the radius $1/2$. Consequently,
\begin{equation}
{\rm vol}\, CP^1=\pi, \qquad \int_{CP^1} \rho = 8\pi\,,\qquad \int_{CP^1} \rho^2 =64\pi\,.
\label{CP14}
\end{equation}
Therefore, (\ref{2Da0}) - (\ref{2Da4}) are consistent with (\ref{CP13}).

The eigenvalues of $\PP_l$ on other complex projective spaces can be found in \cite{Wyss}.

\section*{Acknowledgments}
The author is grateful to Rold\~ao da Rocha for many fruitful discussions. This work was supported in part by CNPq, projects
306208/2013-0 and 456698/2014-0, and FAPESP, project 2012/00333-7.

\end{document}